\documentclass[prl,aps,superscriptaddress,floatfix,onecolumn]{revtex4-2}
\usepackage{graphicx,graphics,psfrag,amsmath,amsthm,calc}
\usepackage{epsfig}
\usepackage{color,comment}
\usepackage{amsfonts}

\usepackage{caption}
\usepackage{amssymb}
\usepackage{bbm}
\usepackage{enumerate}
\usepackage{mathrsfs}
\usepackage[margin=2cm]{geometry}
\usepackage{braket}
\usepackage{tikz}
\usepackage{mathrsfs}
\usepackage [all,arc,curve,color, arrow, matrix,frame]{xy}

\usepackage[colorlinks ,pdfstartview=FitB]{hyperref}

\newtheorem{definition}{Definition}
\newtheorem{corollary}{Corollary}
\newtheorem{prop}{Proposition}
\newtheorem{lemma}{Lemma}

\newcommand{\hU}{\hat{U}}
\newcommand{\hO}{\hat{\mathcal{O}}}

\newcommand{\hW}{\hat{W}}
\newcommand{\hUd}{\hat{U}^\dagger}
\newcommand{\hWd}{\hat{W}^\dagger}
\newcommand{\vth}{\vec{\theta}}
\newcommand{\valp}{\vec{\alpha}}

\def\idty{{\mathchoice {\mathrm{1\mskip-4mu l}} {\mathrm{1\mskip-4mu l}} %
		{\mathrm{1\mskip-4.5mu l}} {\mathrm{1\mskip-5mu l}}}}

\DeclareMathOperator{\tr}{Tr}

\newcommand{\calO}{\mathcal{O}}

\newcommand{\calU}{\mathcal{U}}

\newcommand{\Var}{\operatorname{Var}}

\newcommand{\reportnumber}{FERMILAB-PUB-23-387-SQMS}
\usepackage[angle=0,scale=1,color=black,firstpage=true,opacity=1]{background}

\SetBgContents{\small\texttt{\reportnumber}}
\SetBgPosition{current page.north east}
\SetBgHshift{-3.5cm}
\SetBgVshift{-1.5cm}

\begin{document}

\title{Investigating Parameter Trainability in the SNAP-Displacement Protocol of a Qudit system}

\author{Oluwadara Ogunkoya}
\affiliation{Superconducting and Quantum Materials System Center (SQMS), Batavia, IL, USA.}
\affiliation{Fermi National Accelerator Laboratory, Batavia, IL, USA.}

\author{Kirsten Morris}
\affiliation{Department of Mathematics, University of Nebraska-Lincoln, USA.}

\author{Do\~ga Murat K\"urk\c{c}\"uo\~glu}
\affiliation{Superconducting and Quantum Materials System Center (SQMS), Batavia, IL, USA.}
\affiliation{Fermi National Accelerator Laboratory, Batavia, IL, USA.}
\date{Received: date / Accepted: date}

\begin{abstract}
\noindent In this study, we explore the universality of Selective Number-dependent Arbitrary Phase (SNAP) and Displacement gates for quantum control in qudit-based systems. However, optimizing the parameters of these gates poses a challenging task. Our main focus is to investigate the sensitivity of training any of the SNAP parameters in the SNAP-Displacement protocol. We analyze conditions that could potentially lead to the Barren Plateau problem in a  qudit system and draw comparisons with multi-qubit systems. The parameterized ansatz we consider consists of blocks, where each block is composed of hardware operations, namely SNAP and Displacement gates \cite{fosel2020efficient}. Applying Variational Quantum Algorithm (VQA) with observable and gate cost functions, we utilize techniques similar to those in \cite{mcclean2018barren} and \cite{cerezo2021cost} along with the concept of $t-$design. Through this analysis, we make the following key observations: (a)  The trainability of a SNAP-parameter does not exhibit a preference for any particular direction within our cost function landscape, (b) By leveraging the first and second moments properties of Haar measures, we establish new lemmas concerning the expectation of certain polynomial functions, and (c) utilizing these new lemmas, we identify a general condition that indicates an expected trainability advantage in a qudit system when compared to multi-qubit systems.
\end{abstract}

\maketitle

\section{Introduction} \label{intro}
\noindent Variational Quantum Algorithms (VQAs) have emerged as promising candidates for the implementation on Noisy Intermediate Scale Quantum (NISQ) computers. However, the scalability of VQAs faces a significant obstacle due to the presence of the barren plateau phenomenon. This phenomenon was first introduced by McClean et al. in 2018, where they demonstrated its existence in a wide range of random quantum circuits \cite{mcclean2018barren}. Furthermore, Cerezo et al. investigated the presence of the barren plateau phenomenon even in shallow circuit depths when using global observables, although the effect can be mitigated when employing local observables \cite{cerezo2021cost}. Despite these studies, the literature lacks a comprehensive analysis regarding the potential avoidance of the barren plateau phenomenon by utilizing qudits instead of qubits. In this context, our research investigates the scaling behavior of VQAs concerning the qudit dimension rather than the number of qubits. Rather than focusing on the gradient scaling with respect to the number of qubits, we direct our attention to understanding the impact of increasing the qudit dimension on the VQA's performance. Through our analysis, we aim to provide insights into the trainability and scalability of VQAs using qudits, potentially offering a solution to overcome the barren plateau phenomenon. By exploring this unexplored avenue, we hope to uncover novel strategies for enhancing the efficiency and applicability of VQAs in practical quantum computing tasks. Understanding the implications of qudit-based VQAs can lead to advancements in quantum algorithm design, optimization techniques, and pave the way for the utilization of more powerful quantum systems in various real-world applications.

The objective of this research is to examine how the parameter trainability scales with the increasing qudit dimension in a variational quantum algorithm, focusing on two distinct cost functions. In the context of qubits, certain cost functions have exhibited the barren plateau phenomenon, leading to a significant reduction in trainability, particularly for parameters initialized randomly. This trainability challenge has been observed in both deep and shallow Parametrized Quantum Circuits (PQCs) \cite{mcclean2018barren}, \cite{cerezo2021cost}. Therefore, our investigation aims to ascertain whether similar trainability issues arise when a qudit is utilized instead of multiple qubits. By analyzing the scaling behavior, we aim to gain valuable insights into how the qudit dimension impacts trainability of parameters in the specified cost functions. This exploration is pivotal in understanding the practical advantages and limitations of qudits over qubits within the domain of variational quantum computing. Through our study, we seek to shed light on the comparative efficiency and effectiveness of qudits in quantum algorithms, with a particular focus on trainability. By identifying and understanding potential trainability challenges and limitations specific to qudits, we can devise tailored strategies to overcome these hurdles and enhance the overall performance of variational quantum algorithms, thereby widening their applicability in various quantum computing tasks. Our investigation holds significant implications for the quantum computing community, offering insights into the suitability of qudits in different applications and shedding light on the general challenges associated with trainability in variational quantum algorithms. The outcomes of our research could drive further advancements in the field, aiding the development of robust quantum algorithms and contributing to the realization of practical quantum computing solutions.

\vspace{2mm}
\noindent To achieve universal cavity control, a combination of the Selective Number-dependent Arbitrary Phase (SNAP) gate and Displacement gate set is utilized. According to \cite{krastanov2015universal}, a significant improvement in the scaling of the ansatz length is achieved by alternating the gates in layers of size $\mathcal{O}(d)$, where $d$ represents the Hilbert space dimension, as opposed to the $\mathcal{O}(d^2)$ scaling discussed in \cite{fosel2020efficient}. For a SNAP gate, the execution time is approximately $3-5; \mu$s, while for a Displacement gate, it takes about $0.1 ;\mu$s \cite{axline2018demand}. Consequently, for $d=10$, the total time required to execute the ansatz is approximately $1.1;\mu$s (accounting for $d+1$ Displacement gates) and $30 - 50;\mu$s (for $d$ SNAP gates). The overall time needed for the ansatz ranges from $31.1$ to $51.1;\mu$s, which is significant compared to the cavity's lifetime of $100;\mu$s, considering the additional operations such as readout that need to be performed. While ongoing research aims to further reduce the ansatz length, the focus of this paper is to work with the $\mathcal{O}(d)$-length ansatz. The efficiency of the proposed universal cavity control scheme using SNAP and Displacement gates showcases its potential for various quantum operations. However, continued efforts to optimize the ansatz duration are crucial to accommodate additional tasks and maximize the utility of the cavity within its limited lifetime. By refining the control techniques, this approach can pave the way for enhanced quantum processing and information tasks, contributing to the advancement of quantum technology.

\vspace{2mm}
\noindent $t$-designs are important tools in the field of science,  ranging from classical statistics to quantum information theory. A whole theory of combinatorial $t$-designs can be found in \cite{colbourn2010crc} and \cite{hardin1992new}. Quantum $t$-designs have been studied in \cite{ambainis2007quantum} which are analogous to spherical $2t$-design in $\mathbb{R}^d$. In \cite{dankert2009exact} and \cite{gross2007evenly}, the extension of  spherical $t$-designs to quantum $t$-designs, i.e. averaging polynomials over  $d$- dimensional spheres as compared to averaging polynomial functions of matrices over the set of unitaries (with respect to the Haar measure) was formally introduced. Over a decade before now, interesting results concerning quantum designs have been shown. In \cite{roy2009unitary}, the existence of quantum $t$-designs for every possible combinations of $t$ and $d$ (where $t$ is the polynomial degree and $d$ is the dimension of the underlying space) was proven. For fixed $t$ and $d$, it is quite difficult to explicitly find the designs but there are computational results about the least number of unitaries that make up each quantum $t$-designs: for example, \cite{gross2007evenly} and \cite{roy2009unitary} proved that for a quantum 2-design in and dimension $d$, there are at least $d^4-2d^2+2$  unitaries. The uniqueness of these sets are not guaranteed and finding the quantum $t$-designs with minimal number of unitaries is still a difficult problem \cite{bannai2019explicit} even though much new methods are being suggested \cite{nakata2021quantum}. Another approach is finding how close (up to an error $\epsilon>0$) an ensemble is to a quantum $t$-design. Much of this can be found in \cite{dankert2009exact}, \cite{emerson2005scalable}  and \cite{ambainis2007quantum}. 

\vspace{2mm}
\noindent In section \ref{sec:prelim}, we present a comprehensive collection of precise symbols and definitions. Our proposed ansatz is formulated using hardware blocks, each housing SNAPS and Displacements gates. Additionally, we introduce the concept of trainability, which revolves around assessing the gradient variance of a cost function along specific directions. By delving into this notion, we aim to understand the system's adaptability and its ability to be effectively trained through gradient-based optimization techniques. The combination of detailed notations, hardware-based ansatz, and the exploration of trainability enhances the foundation of our research and opens new possibilities for advancing the field. In section \ref{sec:Barren}, we delve into the study of parameter trainability for diverse cost functions, presenting both analytical and numerical findings. Leveraging the theory of quantum $2-$designs in $1D$, we provide analytical results and bounds for the variance of observable and gate cost function gradients, which are directly linked to the traces of the corresponding observables and gates. Through our analysis, we establish that the trainability condition is intricately governed by these trace functions. To validate our theoretical claims, we offer numerical examples that serve as concrete evidence of the principles put forth in our analytical proofs. These numerical illustrations provide tangible support for the trainability outcomes we have identified, further reinforcing the robustness of our research. A noteworthy highlight of our investigation is the identification of a general condition that leads to achieving trainability advantage specifically in the SNAP-Displacement protocol when compared to the multi-qubit scenario. We elaborate on this condition, demonstrating its potential implications and significance in gradient-based optimization domains. Furthermore, we furnish a specific example that satisfies this condition, showcasing its real-world applicability and relevance. Our approach combines rigorous analytical derivations and practical numerical demonstrations, yielding a comprehensive understanding of parameter trainability in quantum systems. By unveiling the relationships between cost function variance, trace functions, and trainability, our research contributes valuable insights to the field. The implications of our findings can potentially advance quantum computation techniques, paving the way for more efficient and effective quantum algorithms and protocols. Rigorous proofs for the results in section \ref{sec:Barren} are found in the appendix.

\section{Preliminaries and Notation}\label{sec:prelim}
In this section, we draw inspiration from \cite{kurkcuoglu2021quantum} and shift our focus to qudits, moving away from qubits. Specifically, our investigation revolves around a $d$ dimensional qudit space, with parameters $\alpha \in \mathbb{C}$ and $\vec{\theta}=\{\theta_n\}_{n=0}^d$. Each phase $\theta_n$ is chosen uniformly at random from the interval $[0, 2\pi)$. To harness their universality \cite{krastanov2015universal}, we employ the Displacement operator $\hat{D}(\alpha)$ and the Selective Number-dependent Arbitrary Phase (SNAP) gate $\hat{S}(\vec{\theta})$. The displacement operator $\hat{D}(\alpha)$ is crucial in our study. Its definition is as follows: 
\begin{equation}\label{eqn:infinitedisplacement}
    \hat{D}(\alpha) = e^{ \alpha a - \alpha^{\ast}a^{\dagger}},
\end{equation}
where $a$ is the lowering operator and $a^{\dagger}$ is the raising operator such that $a|n\rangle = \sqrt{n}|n-1\rangle$, $a^{\dagger}|n\rangle = \sqrt{n+1}|n+1\rangle$. While the SNAP gate is defined as
\begin{equation}\label{eqn:infiniteSNAP}
    \hat{S}(\vec{\theta}) = \sum_{n=0}^{\infty}e^{i \theta_n} \ket{n}\bra{n}.
\end{equation}
In practice, given the finite dimensional qudit space context we use a truncated SNAP gate defined as 
\begin{equation}\label{eqn:truncatedSNAP}
\hat{S}_d(\vec{\theta}) = \sum_{n=0}^{d}e^{i \theta_n} \ket{n}\bra{n}
\end{equation}
for $\vec{\theta}=\{\theta_n\}_{n=0}^{d}$.

\vspace{2mm}
\noindent By utilizing these powerful and versatile operators, we venture into the realm of qudits, expanding the potential for quantum computation and control. The adoption of parameters randomly chosen from a uniform distribution in the range $[0, 2\pi)$ adds an element of quantum randomness to the setup, which can yield interesting insights into the behavior and performance of qudit-based quantum algorithms. Our approach holds the promise of enhancing the scope and applicability of quantum information processing using qudits. By leveraging the universality of the displacement operator and the SNAP gate, we pave the way for novel quantum computing strategies and explore new possibilities for quantum advantage in various quantum computing tasks. The use of qudits offers a broader state space and additional computational resources, which could be harnessed to tackle more complex problems and drive innovations in quantum technology.

\vspace{3mm}
\noindent Throughout this paper we consider a particular ansatz as given by \cite{fosel2020efficient}:

\begin{equation}\label{eqn:ansatz0}
	U(\vec{\theta},\vec{\alpha})= \hat{B}(\alpha_T,\vec{\theta}_T)\cdots\hat{B}(\alpha_2,\vec{\theta}_2)\hat{B}(\alpha_1,\vec{\theta}_1).
\end{equation}
In this ansatz, $T$ is an integer representing the number of SNAP gates,  $\;\vec{\alpha}=\{\alpha_i\}_{i=1}^T \in \mathbb{R}^T,\;\;\vec{\theta}_j=\{\theta_{kj}\}_{k=1}^d $ where $j=1,\ldots,T$ and $ \theta_{kj}\in [0,2\pi)$. Each term $\hat{B}(\alpha, \vec{\theta})$ is a block of the form 

$$\hat{B}(\alpha, \vec{\theta})=\hat{D}^{\dagger}(\alpha)\hat{S}_d(\vec{\theta})\hat{D}(\alpha).$$
\vspace{2mm}
Consider a PQC $V(\vec{\theta})$ where $C$ is any cost function, e.g. the \textbf{state/observable cost function} 
\begin{equation}
	C_s = 1-\tr(\mathcal{\hat{O}}U(\vec{\theta},\vec{\alpha})\rho_{_0} U^\dagger(\vec{\theta},\vec{\alpha})),\;\;\;\rho_0=\ket{0}\bra{0}
\end{equation}
where $\hat{O}$ is the target  observable, or the \textbf{gate cost function}
\begin{equation}
C_g = 1-|\tfrac{1}{d}\tr[\hU^\dagger_t\hU(\vth,\valp)]|^2
\end{equation}
where $\hU_t$ is the target unitary.
\\\\
In general, a quantum circuit experiences a \textit{barren plateau} if for any parameter $\theta_{\nu}$, there exists some positive real numbers $A,b$ and $n$ such that 

\begin{equation}
\Var\left[\frac{\partial C}{\partial \theta_\nu} \right] = \mathcal{O}(A^{-b n})
\end{equation}
and mean value 
\begin{equation}
\left<\frac{\partial C}{\partial \theta_\nu}\right>=0
\end{equation}
For the multi-qubit case, $A=2$ and $n$ is the number of qubits, while in the single qudit case, $A$ is the dimension of the qudit, and $n=1$.\\ 

\noindent Both McClean et al. and Cerezo et al. and noted the relevance of the Haar measure and $t$-designs in investigating the possible existence of a barren plateau in Parametrized Quantum Circuits (PQC). These notions are also relevant in this paper so we include their definitions below.

\begin{definition}
The Haar measure $\mu$ over the unitary group $\mathcal{U}(d)$ is the uniquely defined left and right invariant measure such that for any unitary matrix $U\in \mathcal{U}(d)$, and $f :\mathcal{U}(d)\rightarrow\mathcal{U}(d)$ we have that
\begin{align}
\int_{\mathcal{U}(d)}f(W)\; d\mu (W) &= \int_{\mathcal{U}(d)} f(UW)\;d\mu (W) \\
&= \int_{\mathcal{U}(d)} f(WU)\;d\mu (W) 
\end{align}

\end{definition}

\begin{definition}
Let $\mathcal{S}$  be an ensemble of unitaries from  $\mathcal{U}(d)$ with a uniform distribution $\nu,$ and  $P_{(t,t)}(U)$ be a polynomial of degree at most $t$ in the matrix elements of both $U$ and $U^{\dagger}$. Then $\mathcal{S}$ is called a {\bf unitary $t$-design} if

\begin{align}
 \int_\mathcal{S} P_{(t,t)}(S)\;d\nu(S) = \int_{\mathcal{U}(d)} P_{(t,t)}(U)\; d\mu (U) .
\end{align}

\end{definition}
\section{Trainability of Parameters for different Cost functions}\label{sec:Barren}
\label{barrenplateau} The following lemmas ( with proofs presented in the appendix section) hold significant importance for evaluating the expectation of operators over a specific set of unitaries in $\calU(d)$. Rather than computing complex integrals over this particular set, we transform the problem into evaluating  integrals (with respect to the Haar measure) over the entire set of unitaries $\calU(d)$. This approach offers two key advantages:

Firstly, the subset of $\mathcal{U}(d)$ over which the integral is to be performed is sometimes vaguely defined, which can introduce uncertainty and challenges in the calculations. By employing the concept of unitary $t-$designs, we can circumvent this ambiguity and achieve a well-defined integration procedure, simplifying the overall evaluation.

Secondly, even when the subset is well-defined, traditional mathematical tools may not be sufficient for performing such complex calculations. The utilization of unitary $t-$designs provides us with a powerful and efficient method to tackle these integrals, facilitating more tractable and manageable evaluations.

\vspace{2mm}
\noindent By leveraging the concept of unitary $t-$designs, we can avoid complications and streamline the evaluation process for these integrals. This approach offers a valuable toolset for handling expectations of operators over unitary sets, enabling more straightforward and effective calculations.
\begin{lemma}\label{lem:11design}
Let an ensemble $\mathcal{S}$ (subset of $\mathcal{U}(d)$)  with a uniform distribution $\nu$ be a unitary $t-$design with $t\geq 1$, and let $C,D:\mathbb{C}^d\rightarrow \mathbb{C}^d$ be linear operators. Then
\begin{equation}
\int_\mathcal{S} \tr[WC]\tr[W^\dagger D] \;d\nu(W)=\int_{\mathcal{U}(d)}\tr[WC]\tr[WD]\;d\mu(W)= \tfrac{1}{d}\tr[CD]
\end{equation}
\end{lemma}

\begin{lemma}\label{lem:22design}
Let an ensemble $\mathcal{S}$ (subset of $\mathcal{U}(d)$)  with a uniform distribution $\nu$ be a unitary $t-$design with $t\geq 2$, and let $C,D,E,F:\mathbb{C}^d\rightarrow \mathbb{C}^d$ be linear operators. Then
\begin{align}
 \int_\mathcal{S}\tr[WC]&\tr[W^\dagger D]\tr[WE]\tr[W^\dagger F] \;d\nu(W)  =\int_{\mathcal{U}(d)}\tr[WC]\tr[W^\dagger D]\tr[WE]\tr[W^\dagger F]\; d\mu(W) \notag\\
  &= \tfrac{1}{d^2-1}\left(\tr[CD]\tr[EF]+ \tr[CF]\tr[ED]\right)-\tfrac{1}{d^2-1}\left(\tr[CDEF]+\tr[CFED]\right)
\end{align}
Moreover, 
\begin{align}
\int\left(\tr[WC]\tr[W^\dagger D]\right)^2\;d\mu(W)= \frac{2}{d^2-1}\left(\tr[CD]\right)^2- \frac{1}{d(d^2-1)}\left(\tr[(CD)^2]\right)^2 
\end{align}

\end{lemma}
 The proofs  for the second equality of the above lemmas is seen in Appendix \ref{app:A} (The first equality follows from the assumption). They are essential to compute the expectation of the cost functions in the later sections.
\subsection{Trainability of Parameters for State/Observable Cost function}
\label{stateobservable}
In this section we investigate the trainability of a parameter within the context of a State/Observable cost function.  In order to compare the `distance' between the targeted state and the engineered state, we consider the State/Observable cost function $C_s$ as defined in the Preliminaries Section, and an observable $\mathcal{\hat{O}}$,
\begin{equation}
	C_s = 1-\tr(\mathcal{\hat{O}}U(\vec{\theta},\vec{\alpha})\rho_{_0} U^\dagger(\vec{\theta},\vec{\alpha})),\;\;\;\rho_0=\ket{0}\bra{0}
\end{equation}
\begin{prop}\label{prop:mean}
The average of the partial derivatives of the cost function $C_s$ with respect to $\theta_\nu$ in a block $\hat{B}(\alpha,\vec{\theta})$ of the ansatz $U(\vth,\valp)$
\begin{equation}\label{eq:MeanStateCost0}
	\braket{\partial_\nu C_s}_{\hU} = 0
\end{equation}
Moreover, 
\begin{equation}\label{eq:StateCost0}
	\braket{(\partial_\mu C_s)^2}_{\hW_A,\hW_B} = \frac{2}{(d-1)(d+1)^2}\left((\tr \hO^2)- \frac{(\tr \hO)^2}{d}\right)
\end{equation}
provided $W_A$ or $W_B$ forms a 2-design.
\end{prop}

A zero expectation for the directional derivative of the cost function implies that no particular direction in the landscape of the cost function is function. i.e., in the landscape of the cost function, the chances of encountering a local minimum or maximum are the same irrespective of the direction. Hence, we can choose any particular direction for a start. Having a zero expectation also reduces the computations to evaluate the variance. In order analyze the chances of  a trainability problem in a particular direction, it suffices to evaluate the variance (in the same direction) by Chebyshev inequality. This shows that the above proposition is a vital result in exploring the possible problems in parameter trainability. The second term on the right hand side of \eqref{eq:StateCost0} hence  controls the trainabilty of the parameter $\theta_\mu$. We give a similar result for a gate cost function later in this paper.
%

For numerical examples, we consider the particular ansatz given in \eqref{eqn:ansatz0} to analyze the trainability of a parameter while considering the fock state $\ket{0}$ as a particular example.

\begin{figure}[h]
\includegraphics[width=0.48\textwidth]{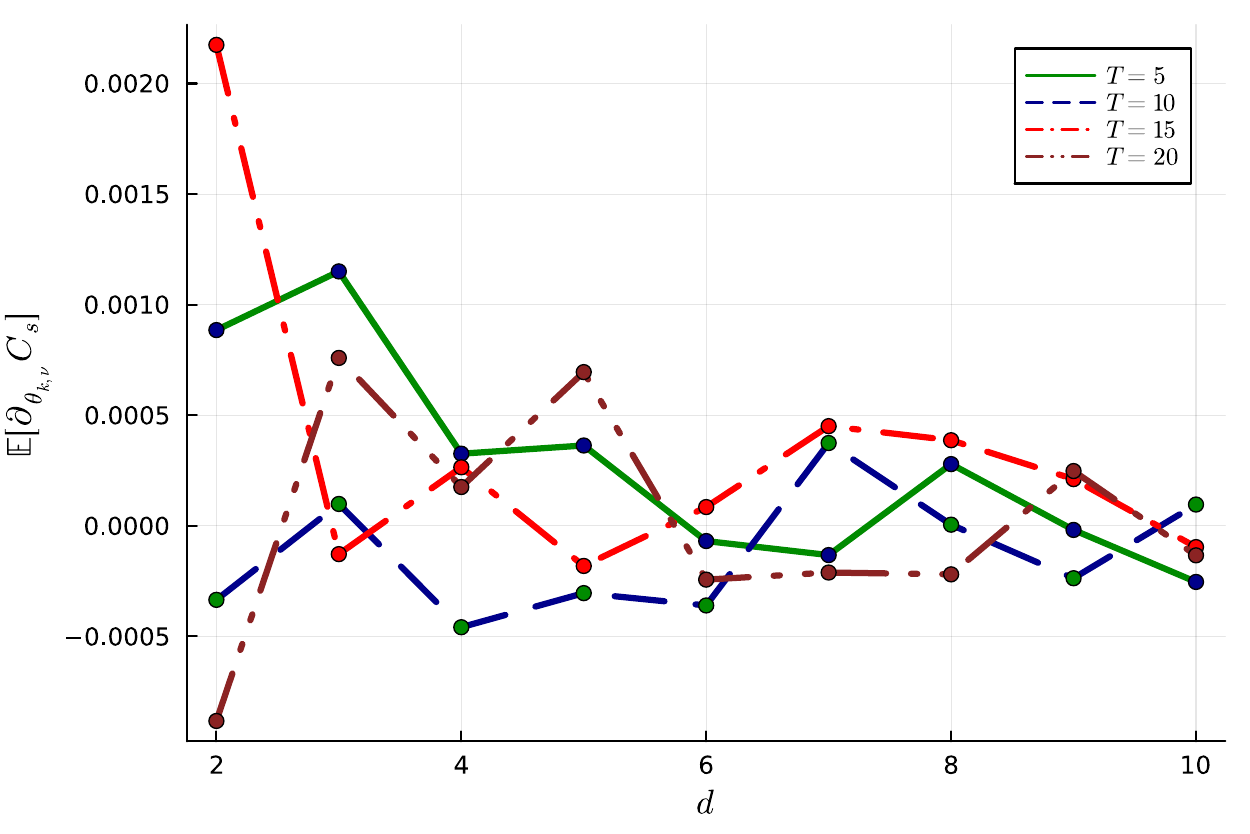}
\includegraphics[width=0.48\textwidth]{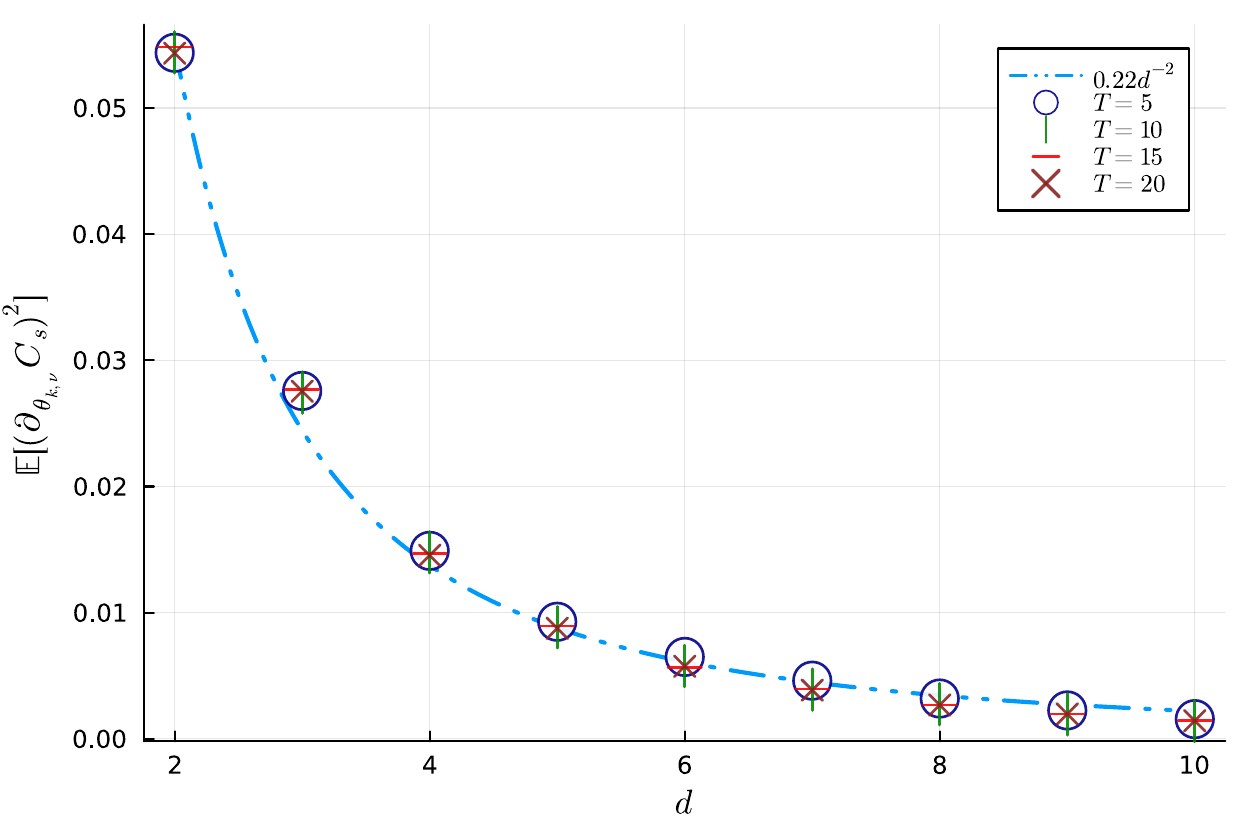}
\caption{Plots showing the mean and the variance of the partial derivative of the State/Observable cost function (with final state $\ket{0}\bra{0}$) with respect to the parameter $\theta_{k,\nu}$ (where $\theta_{k,\nu}$ is the phase associated with  $\ket{\nu}\bra{\nu}$ in the $k-$th block.) For the above plots,  $k=3$, $\nu=1$, and the total number of blocks $T=5,10,15,20$}\label{fig:StateCost}
\end{figure}

\noindent The computation of these variances was carried out with the assumption that the means of the partial derivative of the cost function is zero. It is safe to do this since the analytical results in \eqref{eq:MeanStateCost0} is already validated by Fig. \ref{fig:StateCost}(a). Suppose $\calO = \ket{0}\bra{0}$ in \eqref{eq:StateCost0}. With increase in the dimension of the qudit system, the variance of the partial derivative of the cost function decays quadratically, since  the $\calO$ is a projection and $\tr(\calO^2)-(\tr \calO)^2/d\leq d$. As seen in the Fig. \ref{fig:StateCost}(b), the variance is independent of the number of blocks present in the ansatz  (i.e. for different block sizes $T=5,10,15$), and the polynomial decay is confirmed by numerical results.

For this case $\calO = \ket{0}\bra{0}$, we do not have any advantage of trainability as compared to the the results in \cite{cerezo2021cost}. In fact the quadratic decay rate  according to \eqref{eq:StateCost0} is much worse compared to the decay rate of approximately $0.21$ as seen in [Cor. 1(i), \cite{cerezo2021cost}]. Therefore, the vital question remains: \textit{Are there cases where a qudit system poses trainability advantages over multi-qubits systems}?

\vspace{2mm}
\noindent We will give some general case where we expect qudit trainability advantage over the multi-qubit case. Consider the cases where  $\hat{\calO}$  satisfy conditions in \textit{part (ii)}, \textit{Theorem} 1 of \cite{cerezo2021cost} (i.e.  $\hat{\calO}$ can only be written as an sum of operator tensors: $\hat{\calO} = c_0\idty +\sum^N_{i=1}c_i\hat{\calO}_{i1}\otimes\hat{\calO}_{i2}\otimes\cdots\otimes\hat{\calO}_{i\xi}$,  with  $\tr[\hat{\calO}_{ik}]=0,\;\;\tr[(\hat{\calO}_{ik})^2]\leq 2^m$, $m\geq 2$), then Var$(\partial_\mu C_s)\leq 2^{-an}$ where $a=(1-1/m)$. If $d=2^n$ in \eqref{eq:StateCost0}, and $\tr[\hat{\calO}^2]- 2^{-n}[\tr(\hat{\calO})]^2)> 2^{(3-a)n+1}$ (where $n$ is the number of qubits, $\;a\in (0.5,1))$, then by the  trivial bound $2^n>1]$, we get that Var$(\partial_\mu C_s)> 2^{-an}$. Thus, an advantage over the multi-qubit case.

\vspace{2mm}
\noindent To better comprehend the aforementioned analytical argument, consider introducing the Particle Number operator denoted by $\hat{\calO}$. This operator exhibits non-zero entries exclusively on its diagonal, forming the sequence $(d-1, d-2, \ldots, 0)$. The trace of $\hat{\calO}$, represented as $\tr[\hat{\calO}]$, is given by $(d/2)(d-1)$, while the trace of $(\hat{\calO})^2$ is $\tr[(\hat{\calO})^2] = (d/6)(d-1)(2d-1)$. By leveraging equation \eqref{eq:StateCost0}, we deduce that the variance of $\partial_\mu C_s$, denoted as Var$(\partial_\mu C_s)$, equals $(2d-3)/(3d+3)$. Notably, this function is both convergent and monotonically increasing with respect to the variable $d$. As a result, we can conclude that Var$(\partial_\mu C_s)$ remains bounded as $d$ varies, with $1/9 \leq $Var$(\partial_\mu C_s)\leq 2/3$. This implies that the variability in $\partial_\mu C_s$ does not grow without bound and instead lies within a specific range. Understanding the behavior of Var$(\partial_\mu C_s)$ is crucial in the context of the analytical argument presented. It allows us to gain insights into the particle number distribution and the associated state cost function $C_s$ for different values of $d$. 

\begin{figure}[h]
  \includegraphics[width=0.48\textwidth]{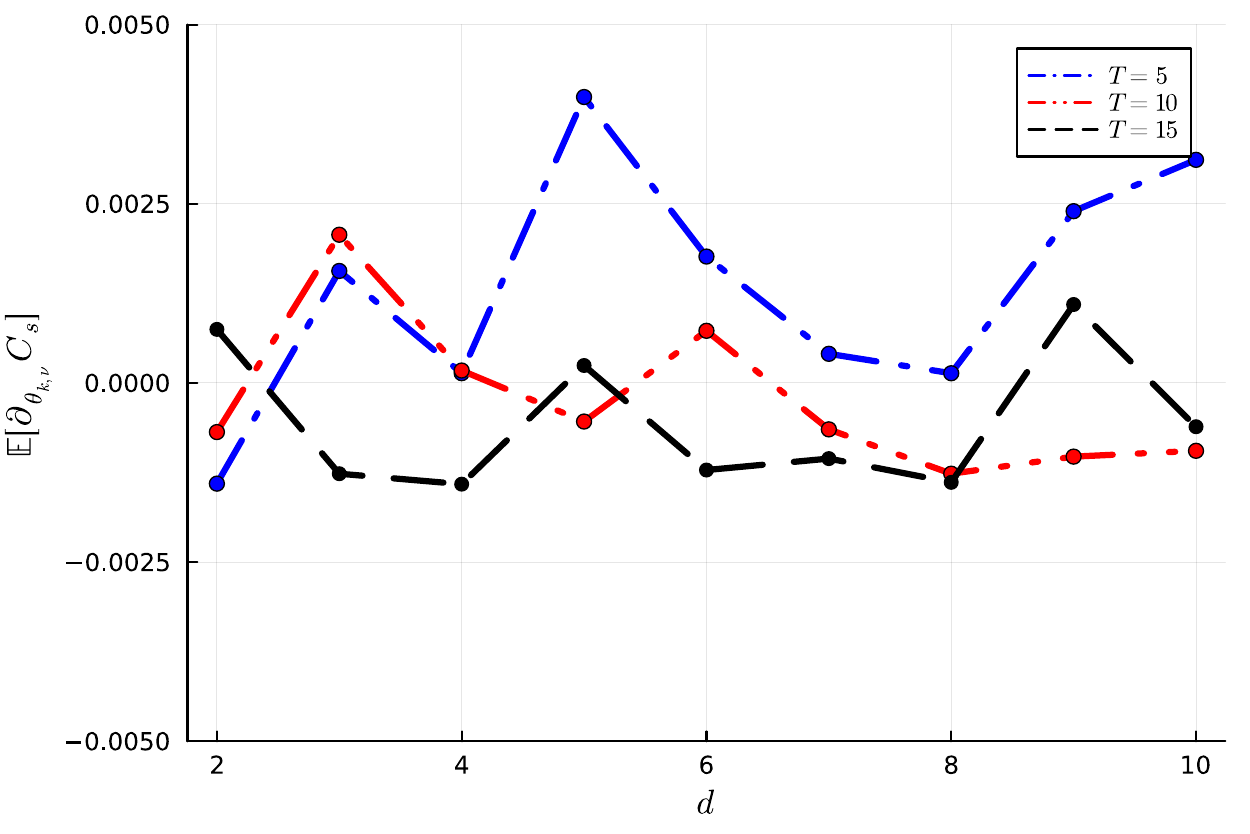}
  \includegraphics[width=0.48\textwidth]{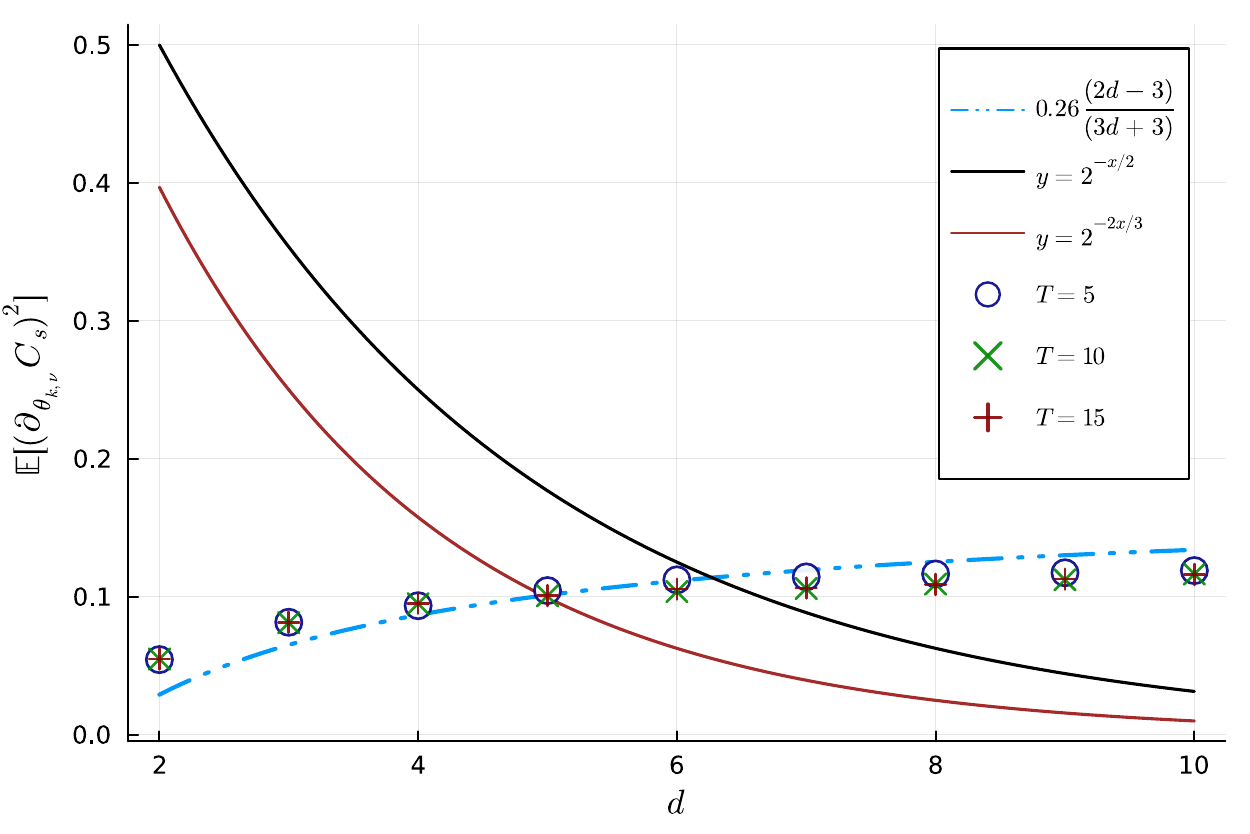}
\caption{Plots showing the mean and the variance  of $\partial_{\theta_{k,\nu}}C_s$ where the observable is the Particle Number Operator. 3(b) combines  plots representing the decay bound as provided in \cite{cerezo2021cost} (with different decay rates $a=0.5$ and $a=0.67$) and the experimental plots from the qudit case.}
\end{figure}
\noindent Observe that in the above plots, trainability of a parameter in the few qubits case (1-2 qubits) is much faster (depending on the decay rate $a$) than training in a qudit case. But as qudit system dimension $d$ gets larger, trainability  becomes better in the qudit case since blue plot is increasing (not too fast though) and the red and black plots decrease. Evidently for observables requiring   large number of qubits, training on a qudit system is of advantage over multi-qubits systems.

\subsection{Trainability of Parameters for Gate Cost function}
\label{gatecost}
In this context, our focus is on examining the trainability concerning the gate cost function. The preparation of a unitary gate using the SNAP-Displacement protocol holds a significance akin to that of preparing a state. Remarkably, if we possess accurate gates, it alleviates the complexities surrounding the preparation of highly excited states. Thus, the quality of unitary gates is crucial for efficient quantum operations. Similar to our approach for the state/observable cost function, we adopt a comparable argument for the gate cost function. This means employing the same partition for the unitary operator $\hU(\vth, \valp)$, as presented in \eqref{eqn:Upart}. The gate cost function, denoted as $C_g$, captures the essence of the cost associated with generating the desired unitary gate and is defined in the preliminaries section \ref{sec:prelim}. Examining the trainability within the realm of the gate cost function allows us to assess the feasibility and efficiency of obtaining the desired unitary transformations through the SNAP-Displacement protocol. By understanding the factors influencing the trainability of gates, we can devise strategies to improve the protocol's performance and enhance the accuracy of unitary operations. This, in turn, contributes to the overall effectiveness of quantum computations and tasks. In summary, the trainability analysis of the gate cost function emphasizes the importance of precise unitary gates in quantum computing applications. It underscores the role of the SNAP-Displacement protocol in generating reliable unitary operations and highlights the potential benefits of having accurate gates, leading to simplified state preparations and improved quantum computation capabilities.

\noindent For simplicity, we write $\hU$ in place of $\hU(\vth,\valp)$ henceforth. A similar result as in \eqref{prop:mean} is provided.
\begin{prop}\label{prop:gatemeanvar}
The average of the partial derivatives of the cost function $C_g$ with respect to $\theta_\nu$ in a block $\hat{B}(\alpha,\vec{\theta})$ of the ansatz $U(\vth,\valp)$
\begin{equation}\label{eq:GateCostmean}
	\braket{\partial_\nu C_g}_{\hU} = 0
\end{equation}
\\
Moreover, 
\begin{align}\label{eq:GateCost1}
\braket{(\partial_\nu C_g)^2}_{\hU}=\frac{2(d^6+d^5-4d^4-3d^3+5d^2+2d-1)}{d^4(d^2-1)^4}+\frac{2}{d^4(d^2-1)^4}(|\tr[\hU_t]|^4-2|\tr[\hU_t]|^2)
\end{align}

\begin{equation}\label{eq:GateCost3}
\braket{(\partial_\nu C_g)^2}_{\hU} = \mathcal{O}(\mathrm{poly}(1/d^6))
\end{equation}
provided $W_A$ or $W_B$ form a 2-design.
\end{prop}
We further explore what happens when we try to train parameters  in trying to obtain certain gates. A particular example is shown in the figure below for the Identity gate. Analytical results in \eqref{eq:GateCostmean} and \eqref{eq:GateCost3} are validated.


\begin{figure}[h]
  \includegraphics[width=0.48\textwidth]{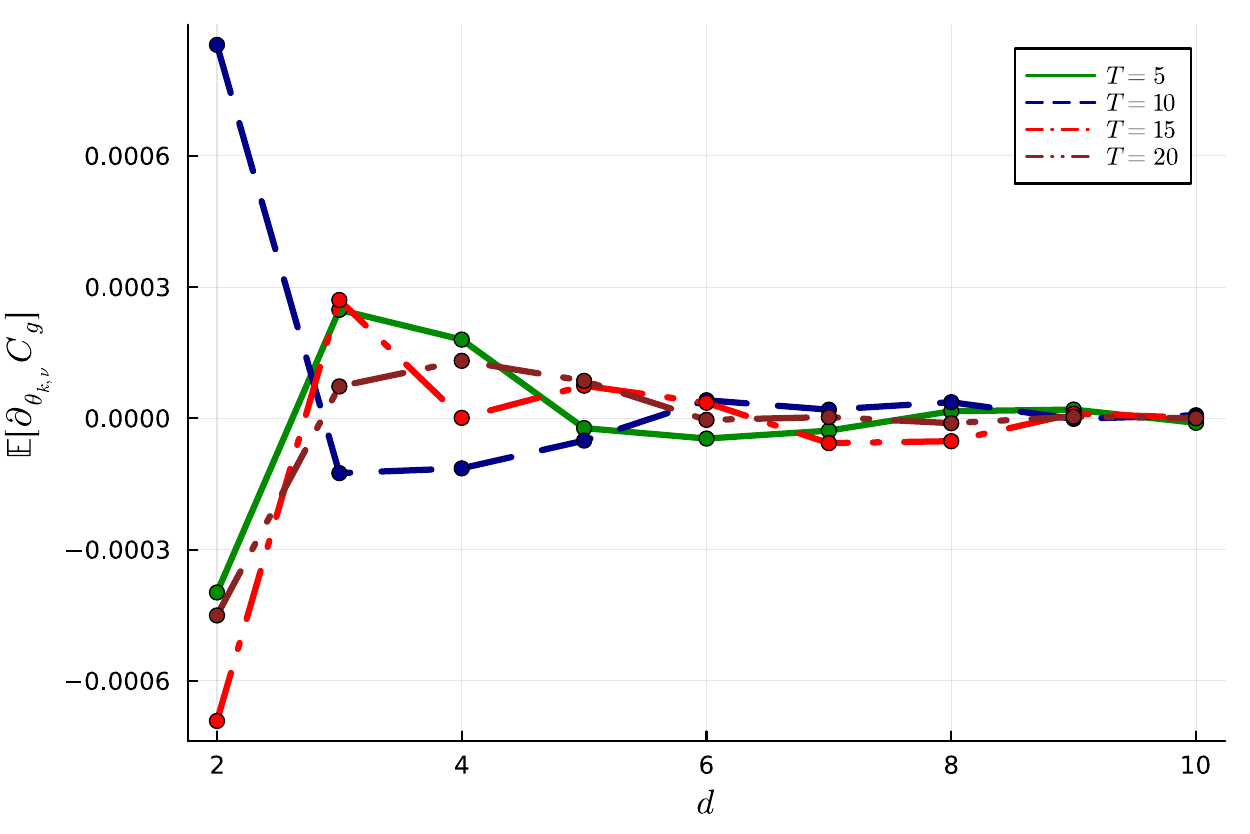}
  \includegraphics[width=0.48\textwidth]{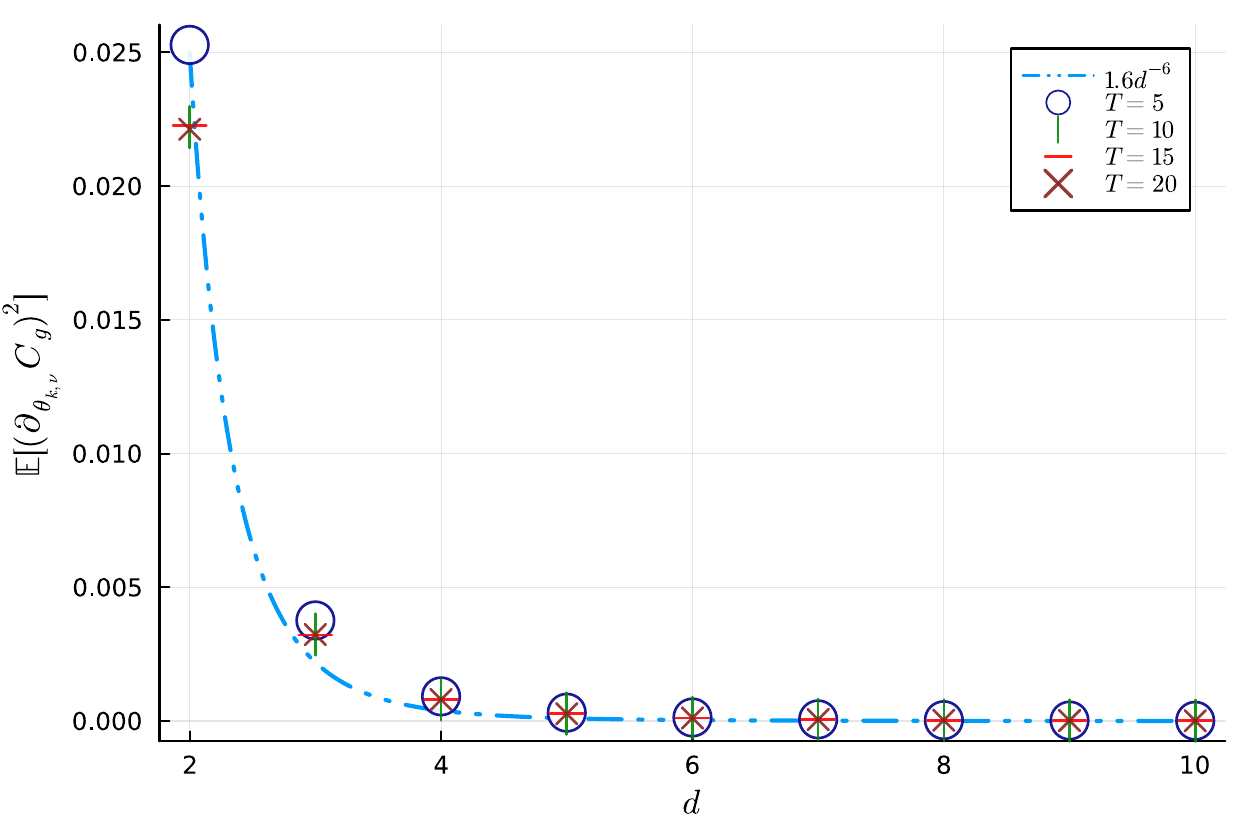}
\caption{Plots showing the mean and the variance of the partial derivative of the Gate cost function (with target gate being the identity map $I_d$) with respect to the parameter $\theta_{k,\nu}$ (where $\theta_{k,\nu}$ is the phase associated with  $\ket{\nu}\bra{\nu}$ in the $k-$th block.) For the above plots,  $k=3$, $\nu=1$, and the total number of blocks $T=5,10,15,20,25$.}
\label{fig:GateCost}
\end{figure}
\noindent The above plots show that the leading order decay  is independent of the number of blocks in the ansatz. Unlike the plots for the state cost function, trainability of parameters in the ansatz for the gate cost function is not dependent on any particular gate. According to equation \eqref{eq:GateCost3}, the variance leading order decay is of degree 6, hence, confirmed by the plots of uniformly distributed variables in Fig. 3(b).  
\\\\
To justify the 2-design assumptions in the above propositions, notice that the plots in Fig.\eqref{fig:GateCost} are plots from means over uniform distribution of unitaries in the sets $\{\hW_A|\; \alpha,\theta_j\in [0,2\pi), \;j=1,\ldots,d\} $ and $\{\hW_B|\; \alpha,\theta_j\in [0,2\pi), \;j=1,\ldots,d\}$. This indicates these sets form 2-designs (and hence 1-designs by definition).

\newpage
\begin{acknowledgments}
\section{Acknowledgments}
\noindent This material is based upon work supported by the U.S. Department of Energy, Office of Science, National Quantum Information Science Research Centers, Superconducting Quantum Materials and Systems Center (SQMS) under contract number DE-AC02-07CH11359. This research was supported in part by an appointment with the National Science Foundation (NSF) Mathematical Sciences Graduate Internship (MSGI) Program sponsored by the NSF Division of Mathematical Sciences. This program is administered by the Oak Ridge Institute for Science and Education (ORISE) through an interagency agreement between the U.S. Department of Energy (DOE) and NSF. ORISE is managed for DOE by ORAU. All opinions expressed in this paper are the author's and do not necessarily reflect the policies and views of NSF, ORAU/ORISE, or DOE.
\end{acknowledgments}

\appendix
\section{Appendix}\label{app:A}
The first and second moments of Haar measures are used to analyze expectations of trace operators \cite{collins2006integration, puchala2011symbolic}. These calculations would have been very difficult since it entails numerous integrations of polynomial functions over the set of unitaries $\calU(d)$, where $d$ is the dimension of the unitary group. The first two moment results are listed below and some of the lemmas from \cite{cerezo2021cost} (based on the these moments) for parameterized circuits with $t$-designs are also applied.

\begin{align}
	\int_{\calU(d)} w_{\bf i_1,j_1}\;w^*_{\bf i_2,j_2}\;d\mu(W)&= \frac{\delta_{\bf i_1,i_2}\delta_{\bf j_1,j_2}}{d},\label{eqn:firstmo}\\\notag\\
	\int_{\calU(d)} w_{\bf i_1,j_1}w_{\bf i_2,j_2}w^*_{\bf p_1,k_1}w^*_{\bf p_2,k_2}\;d\mu(W)&=\frac{1}{d^2-1}(\delta_{\bf i_1, p_1}\delta_{\bf i_2, p_2}\delta_{\bf j_1, k_1}\delta_{\bf j_2, k_2}+ \delta_{\bf i_i, p_2}\delta_{\bf i_2, p_1}\delta_{\bf j_1, k_2}\delta_{\bf j_2, k_1})\label{eqn:secondmo}\\\notag
	&-\frac{1}{d(d^2-1)}(\delta_{\bf i_1, p_1}\delta_{\bf i_2, p_2}\delta_{\bf j_1, k_2}\delta_{\bf j_2, k_1}+ \delta_{\bf i_i, p_2}\delta_{\bf i_2, p_1}\delta_{\bf j_1, k_1}\delta_{\bf j_2, k_2})
\end{align}
The proofs to Lemmas \ref{lem:1design} and \ref{lem:2design}  below can be found in \cite{cerezo2021cost} while we give the proofs of some new lemmas.
\begin{lemma}\label{lem:1design}
Let an ensemble $\mathcal{S}$ (subset of $\mathcal{U}(d)$)  with a uniform distribution $\nu$ be a unitary $t-$design with $t\geq 1$, and let $A,B:\mathbb{C}^d \rightarrow \mathbb{C}^d$ be arbitrary linear operators. Then
\begin{equation}
\int_\mathcal{S} \tr[WAW^\dagger B]\;d\nu(W)	 =\int_{\mathcal{U}(d)}\tr[WAW^\dagger B] \;d\mu(W)= \frac{\tr[A]\tr[B]}{d}
	\end{equation}
\end{lemma}
\begin{lemma}\label{lem:2design} 
Let an ensemble $\mathcal{S}$ (subset of $\mathcal{U}(d)$)  with a uniform distribution $\nu$ be a unitary $t-$design with $t\geq 2$, and let $A,B,C,D:\mathbb{C}^d\rightarrow \mathbb{C}^d$ be linear operators. Then
\begin{align}
\int_\mathcal{S} \tr[WAW^\dagger BWCW^\dagger D] \;d\nu(W)
&=\int_{\calU(d)}\tr[WAW^\dagger BWCW^\dagger D] \;d\mu(W)\notag\\
	&= \frac{1}{d^2-1}\left(\tr[A]\tr[C]\tr[BD]+\tr[AC]\tr[B]\tr[D]\right)\\
	&-\frac{1}{d(d^2-1)}\left(\tr[AC]\tr[BD]+\tr[A]\tr[B]\tr[C]\tr[D]\right)\notag
	\end{align}
\end{lemma}
\begin{lemma}
 Let an ensemble $\mathcal{S}$ (subset of $\mathcal{U}(d)$)  with a uniform distribution $\nu$ be a unitary $t-$design with $t\geq 2$ and let $A,B,C,D:\mathbb{C}^d\rightarrow \mathbb{C}^d$ be linear operators. Then
 \begin{align}
 \int_{\mathcal{S}}\tr[WAW^\dagger B]\tr[WCW^\dagger D] \;d\nu(W)&=\int_{\calU(d)}\tr[WAW^\dagger B]\tr[WCW^\dagger D] \;d\mu(W)\notag\\
 	&= \frac{1}{d^2-1}\left(\tr[A]\tr[B]\tr[C]\tr[D]+\tr[AC]\tr[BD]\right)\\
 	&-\frac{1}{d(d^2-1)}\left(\tr[AC]\tr[B]\tr[D]+\tr[A]\tr[C]\tr[BD]\right)\notag
 \end{align}
\end{lemma} 
\begin{corollary}\label{cor:quadratic} Suppose $A,B$ satisfy the assumptions in the above lemma,
\begin{align}
\int_{\mathcal{S}} (\tr[WAW^\dagger B])^2\;d\nu(W)&= \int_{\calU(d)} (\tr[WAW^\dagger B])^2 \;d\mu(W)\\\notag
		&=\frac{1}{d^2-1}\left[(\tr[A])^2\left((\tr[B])^2-\frac{\tr[B^2]}{d}\right)+ \tr[A^2]\left(\tr[B^2]-\frac{1}{d}(\tr[B])^2\right)\right]
\end{align}
\begin{proof}
	Setting $A=C$ and $B=D$ in the above lemma and combining similar terms, we get the desired result.
\end{proof}
\end{corollary}

\begin{proof}[Proof of Lemma \ref{lem:11design}:]

By expressing the trace and by applying \eqref{eqn:firstmo}, we have
\begin{align}
\int_{\calU(d)} \tr[WC]\tr[W^\dagger D] \;d\mu(W)&=\sum_{i,j,k,m}^d c_{ji}d_{mk}\int_{\calU(d)} w_{ij} w^*_{mk} \;d\mu(W)\\
&=\frac{1}{d}\sum_{i,j}^dc_{ji}d_{ij}= \frac{1}{d}\tr[CD]
\end{align}
\end{proof}

\begin{proof}[Proof of Lemma \ref{lem:22design}:]
By expressing the trace and applying \eqref{eqn:secondmo},
\begin{align}
\int_{\calU(d)}\tr[WC]\tr[W^\dagger D]&\tr[WE]\tr[W^\dagger F]\; d\mu(W) = \int_{\calU(d)} \sum_{\substack{i,j,k,m\\ i',j',k',m'}}^d c_{ji}d_{mk}e_{j'i'}f_{m'k'}w_{ij}w_{mk}^*w_{i'j'}w^*_{m'k'}\;d\mu(W)\\
&=\sum_{\substack{i,j,k,m\\ i',j',k',m'}}^d c_{ji}d_{mk}e_{j'i'}f_{m'k'}\left[\frac{1}{d^2-1}(\delta_{im}\delta_{i'm'}\delta_{jk}\delta_{j'k'}+\delta_{im'}\delta_{i'm}\delta_{jk'}\delta_{j'k})\right.\\\notag
&\hspace{4.5cm}\left.-\frac{1}{d(d^2-1)}(\delta_{im}\delta_{i'm'}\delta_{jk'}\delta_{j'k}+\delta_{im'}\delta_{i'm}\delta_{jk}\delta_{j'k'})\right]\\
\end{align}
By summing over the indices $k,m,k',m'$, the last equality yields
\begin{align}
\sum_{\substack{i,j\\i',j'}}^d c_{ji}e_{j'i'}\left[\frac{1}{d^2-1}(d_{ij}f_{i'j'}\right.&\left.+d_{i'j'}f_{ij})-\frac{1}{d(d^2-1)}(f_{i'j}d_{ij'}+d_{i'j}f_{ij'})\right]\\\notag
&=\frac{1}{d^2-1}(\tr[CD\tr[EF]+\tr[CF]\tr[ED])-\frac{1}{d(d^2-1)}(\tr[CDEF]+\tr[CFED])
\end{align}
\\
Setting $C=E$ and $D=F$ proves the last part of the lemma. \qed
\end{proof}
Following the notations in  section \ref{sec:prelim}, where  the  parameterized ansatz $U(\vec{\theta},\vec{\alpha})$ given by 
\begin{equation}\label{eqn:ansatz}
	U(\vec{\theta},\vec{\alpha})= \hat{B}(\alpha_T,\vec{\theta}_T)\cdots\hat{B}(\alpha_2,\vec{\theta}_2)\hat{B}(\alpha_1,\vec{\theta}_1)
\end{equation}
and\;\; $\hat{B}(\alpha,\vec{\theta})=\hat{D}^\dagger(\alpha)\hat{S}_N(\vec{\theta})\hat{D}(\alpha)$, we can then write $\hat{U}(\vth,\valp)= \hat{U_R}\hat{B}\hat{U_L}$ where
\begin{align}\label{eqn:Upart}
	\hat{U_R}= &\hat{B}(\alpha_T,\vec{\theta}_T)\cdots\hat{B}(\alpha_{\nu-1},\vec{\theta}_{\nu-1}),\\\notag
\hat{U_L}&=\hat{B}(\alpha_{\nu+1},\vec{\theta}_{\nu+1})\cdots \hat{B}(\alpha_{1},\vec{\theta}_{1})\;\;\;\;\text{and} \;\;\hat{B}(\alpha_{\nu},\vec{\theta}_{\nu}) = \hat{B}
\end{align}
\\
From the definition of the SNAP gate,
\begin{equation}\label{eqn:block}
	\hat{B}(\vth,\valp) =\underset{\hat{W}_A}{\underbrace{\hat{D^\dagger}(\alpha)\prod_{j=0}^{\nu-1}e^{i\theta_j\ket{j}\bra{j}}e^{i\theta_\nu\ket{\nu}\bra{\nu}}}}\;\underset{\hat{W}_B}{\underbrace{\prod_{j=\nu+1}^d e^{i\theta_j\ket{j}\bra{j}}D(\alpha)}}
\end{equation}
Hence, $\hat{B}(\vth,\valp)=\hat{W}_A\hat{W}_B$. For ease of writing we write $B$ in  place of the block $\hat{B}(\vth,\valp)$. Taking the partial derivative of $B$ with respect to $\nu$, i.e.,

\begin{equation}\label{eqn:bderiv}
	\partial_\nu B\equiv \frac{\partial B}{\partial\theta_\nu}=i\hW_A\rho_\nu \hW_B \;\;\;\;\text{where}\;\;\;\rho_\nu = \ket{\nu}\bra{\nu}.
\end{equation}

\begin{proof}[Proof  of Proposition \ref{prop:mean}:]
\begin{align}
	\partial_\nu  C_s &= -\tr[\hO(\partial_\nu\hat{U}\rho_0\hat{U}^\dagger+\hat{U}\rho_0\partial_\nu\hat{U}^\dagger)]\\\notag\\
	&= -i\tr(\hO\hU_R\hW_A\rho_\nu\hW_B\hU_L\rho_0\hUd_L\hWd_B\hWd_A\hUd_R-\hO\hU_R\hW_A\hW_B\hU_L\rho_0\hUd_L\hWd_B\rho_\nu\hWd_A\hUd_R)\\\notag\\
	&=-i\tr[\hWd_A\hUd_R\hO\hU_R\hW_A(\rho_\nu\hW_B\hU_L\rho_0\hUd_L\hWd_B - \hW_B\hU_L\rho_0\hUd_L\hWd_B\rho_\nu)]\\\notag\\
	&=-i\tr(\hWd_A\hUd_R\hO\hU_R\hW_A[\rho_\nu, \hW_B\hU_L\rho_0\hUd_L\hWd_B])
\end{align}
Similarly, 
\begin{equation}
	\partial_\nu  C_s = -i\tr([\hWd_A\hUd_R\hO\hU_R\hW_A,\rho_\nu]\hW_B\hU_L\rho_0\hUd_L\hWd_B) \label{eq:partial2}
\end{equation}
To compute $\braket{\partial_\mu  C_s}$, there is need to compute expectation over $\hW_A, \hW_B,\hU_L,\hU_R$. We then consider the three cases (1) $\hW_A$ is a 1-design, (2) $\hW_B$ is a 1-design, (3) $\hW_A$ and $\hW_B$ are both 1-designs. \\\\
For the case (1) and using \eqref{eq:partial2},
\begin{align}
	\braket{\partial_\mu  C_g}_{\hW_A}&= -i\tr\left(\left[\int \hWd_A\hUd_R\hO\hU_R\hW_A \;\;d\mu(\hW_A),\rho_\nu\right]\hW_B\hU_L\rho_0\hUd_L\hWd_B\right)\\
	&=-i\tr\left(\left[\frac{\tr(\hUd_R\hO\hU_R)\tr(\idty)}{d},\rho_\nu\right]\hW_B\hU_L\rho_0\hUd_L\hWd_B\right)\\
	&=0
\end{align}
where we have applied Lemma \ref{lem:1design}. Similarly, $\braket{\partial_\mu  C_s}_{\hW_B}=0$. Hence for the three cases, 

\begin{equation}\label{eq:MeanStateCost}
\braket{\partial_\mu  C_s}_{\hU} =0
\end{equation}
To evaluate the variance, it suffices to only consider the expectation of the squared of the gradient.

\begin{align}
	(\partial_\mu  C_s)^2 = -\tr\left[\hWd_A\hUd_R\hO\hU_R\hW_A[\rho_\nu,\hW_B\hU_L\rho_0\hUd_L\hWd_B]\right]^2
\end{align}

Therefore by setting $\hat{A}= [\rho_\nu,\hW_B\hU_L\rho_0\hUd_L\hWd_B]$, and $\hat{B}= \hUd_R\hO\hU_R$, and applying Corollary \ref{cor:quadratic},
\begin{align}
	\braket{(\partial_\mu  C_s)^2}_{\hW_A} &= \frac{1}{d^2-1}\left[(\tr \hat{A})^2\left((\tr \hat{B})^2-\frac{\tr \hat{B}^2}{d}\right)\right] + (\tr \hat{A}^2)\left((\tr \hat{B}^2)-\frac{(\tr \hat{B})^2}{d}\right)\\\notag\\
	&= \frac{1}{d^2-1}\left[(\tr \hat{A}^2)\left((\tr \hat{B}^2)-\frac{(\tr \hat{B})^2}{d}\right)\right]\\\notag\\
	&= \frac{1}{d^2-1}\left[(\tr \hat{A}^2)\left((\tr \hO^2)-\frac{(\tr \hO)^2}{d}\right)\right]
\end{align}
where we have used the fact that $\tr[\hat{A}]=0$ in the second in the second equality. 

We simplify $\tr[\hat{B}^2]$ to have 
\begin{equation}
	\tr \hat{B}^2= 2\left(\tr[\hW_B\hU_L\rho_0\hUd_L\hWd_B\rho_\nu]\right)^2- 2\tr[\hW_B\hU_L\rho_0\hUd_L\hWd_B\rho_\nu]
\end{equation}	
Applying Corollary \ref{cor:quadratic},
\begin{equation}
\int  \left(\tr[\hW_B\hU_L\rho_0\hUd_L\hWd_B\rho_\nu]\right)^2\;d\mu(\hW_B) = \frac{1}{d(d+1)}
\end{equation}
and by Lemma \ref{lem:1design}, we have that 
\begin{equation}
	\int \tr\left(\hW_B\hU_L\rho_0\hUd_L\hWd_B\rho_\nu\right)\;d\mu(\hW_B) = \frac{1}{d}
\end{equation}
Hence,
\begin{equation}\label{eq:StateCost}
	\braket{(\partial_\mu  C_s)^2}_{\hW_A,\hW_B} = \frac{2}{(d-1)(d+1)^2}\left((\tr \hO^2)- \frac{(\tr \hO)^2}{d}\right)
\end{equation}

\end{proof}

\begin{proof}[Proof of Proposition \ref{prop:gatemeanvar}:]

\begin{align}
\partial_\nu C_g&= -\tfrac{1}{d^2}(\tr[\hU_t\hUd_f]\partial_\nu(\tr[\hUd_t\hU_f])+ \text{c.c})\\\notag\\
&= -\frac{1}{d^2}(\tr[\hU_t\hUd_f]\tr[\hUd_t\hU_R(\partial_\nu B)\hU_L]+ \text{c.c})\\\notag\\
&= -\frac{i}{d^2}(\tr[\hU_t\hUd_f]\tr[\hUd_t\hU_R\hW_A\rho_\nu\hW_B\hU_L] + \text{c.c})\\\notag\\
&= -\frac{i}{d^2}(\tr[\hU_t\hUd_L\hWd_B\hWd_A\hUd_R]\tr[\hUd_t\hU_R\hW_A\rho_\nu\hW_B\hU_L] -\tr[\hUd_t\hU_L\hW_B\hW_A\hU_R]\tr[\hU_t\hUd_R\hWd_A\rho_\nu\hWd_B\hUd_L])
\end{align}
 where we have used \eqref{eqn:bderiv} in third equation.
\\\\
 By setting $C_1 = \rho_\nu\hW_B\hU_L\hUd_t\hU_R,\;\;D_1=\hUd_R\hU_t\hUd_L\hWd_b,\;\; C_2= \hU_R\hUd_t\hU_L\hW_B,\;\;D_2= \rho_\nu\hWd_B\hUd_L\hU_t\hUd_R $, it is obvious that $C_1= \rho_\nu D_1^\dagger$ and  $D_2=\rho_\nu D^\dagger_2$. By applying lemma \ref{lem:11design}, we have that 
 \begin{align}
 \braket{\partial_\nu C_g}_{\hW_A} &=-\frac{i}{d^3}(\tr[C_1D_1]-\tr[D_2^\dagger C_2^\dagger])\\\notag\\
 &= -\frac{i}{d^3}(\tr[\rho_\nu D^\dagger_1 D_1]- \tr[D^\dagger_2 D\rho_\nu]) =0
 \end{align}
  A similar result holds for $\braket{\partial_\nu  C_g}_{\hW_B}$. Therefore, we have a zero expectation if  either $\hW_A$ or $\hW_B$, (or both)  is a 1-design.
\\\\
The variance is hence calculated directly as the second moment of the gradient since we have zero expectation.

\begin{align}
(\partial_\nu C_g)^2=-\frac{1}{d^4}\{(\tr[\hW_AC_1]\tr[\hWd_AD_1])^2-2\tr[\hW_AC_1]\tr[\hWd_AD_1]\tr[\hW_AC_2]\tr[\hWd_AD_2] + (\tr[\hW_AC_2]\tr[\hWd_AD_2])^2\} \label{eqn:tr}
\end{align}
 We then consider each terms  separately and apply Lemma \ref{lem:22design}.
\begin{align}
\braket{(\tr[\hW_AC_1]\tr[\hWd_AD_1])^2}_{\hW_A} &= \frac{2}{d^2-1}(\tr[C_1D_1])^2 - \frac{1}{d(d^2-1)}(\tr[(C_1D_1)^2])^2 \\
&= \frac{2}{d^2-1}(\tr[\rho_\nu])^2- \frac{1}{d(d^2-1)}(\tr[\rho_\nu])^2\\
&= \frac{2d-1}{d(d^2-1)}\label{eqn:tr1}
\end{align} 
\\\\

\begin{align}
&\braket{\tr[\hW_A C_1]\tr[\hWd_A D_1]\tr[\hW_A C_2]\tr[\hWd_A D_2]}_{\hW_A} = 
 \frac{1}{d^2-1}(\tr[C_1 D_1]\tr[C_2 D_2]+ \tr[C_2 D_1]\tr[C_1 D_2])\notag\\
&\hspace{8.5cm} -\frac{1}{d(d^2-1)}(\tr[C_1 D_1 C_2 D_2]+ \tr[C_1 D_2 C_2 D_1])\\
 &\hspace{2.5cm}=\frac{1}{d^2-1}((\tr[\rho_\nu])^2 + \tr[C_2 D_1]\tr[C_1 D_2])-\frac{1}{d(d^2-1)}(\tr[\rho_\nu C_2 \rho_\nu C^\dagger_2]+ \tr[C_1 D_2 C_2 D_1])\\
 &\hspace{2.5cm}= \frac{1}{d^2-1}+ \frac{1}{d^2-1}\tr[C_2 D_1]\tr[C_1 D_2]-\frac{1}{d(d^2-1)}(\tr[\rho_\nu C_2 \rho_\nu  C^\dagger_2]+ \tr[C_1 D_2 C_2 D_1])\label{eqn:tr2}
\end{align}

\begin{align}
\braket{(\tr[\hW_AC_2]\tr[\hWd_AD_2])^2}_{\hW_A}& = \frac{2}{d^2-1}(\tr[C_2D_2])^2 - \frac{1}{d(d^2-1)}(\tr[(C_2D_2)^2])^2 \\
&= \frac{2}{d^2-1}(\tr[\rho_\nu])^2- \frac{1}{d(d^2-1)}(\tr[\rho_\nu])^2\\
&= \frac{2d-1}{d(d^2-1)} \label{eqn:tr3}
\end{align}
Substituting equations \eqref{eqn:tr1}, \eqref{eqn:tr2} and \eqref{eqn:tr3} into \eqref{eqn:tr}, we have that
\begin{align}
{\braket{(\partial_\nu C_g)^2}}_{\hW_A}= -\frac{2}{d^5(d+1)}+\frac{2}{d^4(d^2-1)}\tr[C_2D_1]\tr[C_1D_2]- \frac{2}{d^5(d^2-1)}(\tr[\rho_\nu C_2\rho_\nu  C^\dagger_2]+ \tr[C_1D_2C_2D_1])
\end{align}
Next we take expectation of the non-constant terms with respect to $\hW_B$. Starting with the last two terms,
\begin{align}
\braket{\tr[\rho_\nu C_2\rho_\nu C^\dagger_2]}_{\hW_B}& = \braket{\tr[\hW_B\rho_\nu\hWd_B\hUd_L\hU_t\hUd_R\rho_\nu\hU_R\hUd_t\hU_L]}_{\hW_B}\\
&= \tfrac{1}{d}\tr[\rho_\nu]\tr[\hUd_L\hU_t\hUd_R\rho_\nu\hU_R\hUd_t\hU_L]\\
&= \tfrac{1}{d}
\end{align}

\begin{align}
\braket{\tr[C_1D_2C_2D_1]}_{\hW_B}&=\braket{\hW_B\hU_L\hUd_t\hU_R\rho_\nu\hUd_R\hU_t\hUd_L\hWd_B\rho_\nu}_{\hW_B}\\
&= \tfrac{1}{d}\tr[\hU_L\hUd_t\hU_R\rho_\nu\hUd_R\hU_t\hUd_L]\\
&= \tfrac{1}{d}
\end{align}
From the definition of $C_1,\; C_2, \;D_1$ and  $D_2$, we can re-write 
\begin{equation}
\tr[C_2D_1]=\tr[\hW_B C_3\hWd_B D_3]\qquad \text{and} \qquad \tr[C_1D_2]= \tr[\hW_BC_3^\dagger\hWd_B D^\dagger_3]
\end{equation}
where 
\begin{equation}
C_3 = \hUd_R\hU_t\hUd_L \qquad  \text{and} \qquad D_3 = \hU_R\hUd_t\hU_L \qquad
\end{equation}
Therefore, by lemma \ref{lem:2design},
\begin{align}
&\braket{\tr[\hW_B C_3\hWd_B D_3]\tr[\hW_BC_4\hWd_B D_4]}_{\hW_B}=\frac{1}{d^2-1}(\tr[C_3]\tr[D_3]\tr[C^\dagger_3]\tr[D^\dagger_3]+ d^2)\notag\\
&\hspace{8cm}-\frac{1}{d(d^2-1)}(d\tr[D_3]\tr[D^\dagger_3]+d\tr[C_3]\tr[C^\dagger_3])\\\notag\\
&\hspace{2cm}= \frac{1}{d^2-1}\tr[C_3]\tr[C^\dagger_3]\tr[D_3]\tr[D^\dagger_3]+ \frac{d^2}{d^2-1}-\frac{1}{d^2-1}(\tr[D_3]\tr[D^\dagger_3]+\tr[C_3]\tr[C^\dagger_3])
\end{align}
Hence,
\begin{align}\label{eqn:WAWB}
\braket{(\partial_\nu C_g)^2}_{\hW_A, \hW_B}& = \frac{2(d^3-d^2-d+2)}{d^6(d^2-1)^2}+\frac{2}{d^4(d^2-1)^2}(\tr[C_3]\tr[C^\dagger_3]\tr[D_3]\tr[D^\dagger_3]-\tr[D_3]\tr[D^\dagger_3]-\tr[C_3]\tr[C^\dagger_3])
\end{align}
From the construction of $\hU(\vth,\valp)$ as an alternating layer in equation \eqref{eqn:ansatz}, we find the expectation of $(\partial_\nu C_g)^2$ with respect to $\hU_R$ ( and in fact $\hU_L$) by considering the  each of the blocks composed in $\hU_R$ (respectively in $\hU_L$) since each of the blocks are independent. Recall  from \eqref{eqn:block} that each block is of the form $\hat{B}=\hW_A\hW_B$.  
\\\\
Without loss of generality, set the first block in $\hU_R$ as $\hat{B}_1=\hW_{A_1}\hW_{B_1}$ such that $\hU_R=\hat{B_1}\hU_{R_1}= \hW_{A_1}\hW_{B_1}\hU_{R_1}$. Note that $\hU_{R_1}$ denotes  $\hU_R$ without the first (parameterized) block. Hence 
\begin{equation}
C_3 = \hUd_{R_1}\hWd_{B_1}\hWd_{A_1}\hU_t\hUd_L\qquad \text{and} \qquad D_3 = \hW_{A_1}\hW_{B_1}\hU_{R_1}\hUd_t\hU_L
\end{equation}
Now we take expectation (with respect to $\hW_{A_1}$ and $\hW_{B_1}$) of the last three terms of \eqref{eqn:WAWB}.

\begin{align}
\braket{\tr[D_3]\tr[D^\dagger_3]}_{\hW_{A_1}}&= \braket{\tr[\hW_{A_1}\hW_{B_1}\hU_{R_1}\hUd_t\hU_L]\tr[\hWd_{A_1}\hUd_L\hU_t\hUd_{R_1}\hWd_{B_1}]}_{\hW_{A_1}}\\
&=\frac{1}{d}\tr[\hW_{B_1}\hU_{R_1}\hUd_t\hU_L\hUd_L\hU_t\hUd_{R_1}\hWd_{B_1}]= 1
\end{align}
where we have applied property of trace in the first equality and lemma \ref{lem:11design}
 in the second equation.
\\\\
Similarly, 
\begin{equation}
\braket{\tr[C_3]\tr[C^\dagger_3]}_{\hW_{A_1}}= \braket{\tr[\hWd_{A_1}\hU_t\hUd_L\hUd_{R_1}\hWd_{B_1}]\tr[\hW_{A_1}\hW_{B_1}\hU_{R_1}\hU_L\hUd_t]}_{\hW_{A_1}}= 1
\end{equation}

\begin{align}
\braket{\tr[C_3]\tr[C^\dagger_3]\tr[D_3]\tr[D^\dagger_3]}_{\hW_{A_1}}&= \langle \tr[\hW_{A_1}\hW_{B_1}\hU_{R_1}\hU_L\hUd_t]\tr[\hWd_{A_1}\hU_t\hUd_L\hUd_{R_1}\hWd_{B_1}]\notag\\
&\hspace{3cm}\times \tr[\hW_{A_1}\hW_{B_1}\hU_{R_1}\hUd_t\hU_L]\tr[\hWd_{A_1}\hUd_L\hU_t\hUd_{R_1}\hWd_{B_1}]\rangle_{\hW_{A_1}}\\
& = \frac{1}{d^2-1}(d^2 + \tr[\hW_{B_1}\hU_{R_1}\hUd_t\hU_L\hU_t\hUd_L\hUd_{R_1}\hWd_{B_1}]\tr[\hW_{B_1}\hU_{R_1}\hU_L\hUd_t\hUd_L\hU_t\hUd_{R_1}\hWd_{B_1}])\notag\\
&\hspace{3cm}-\frac{1}{d(d^2-1)}(2d)\\
&= \frac{d^2-2}{d^2-1} +\frac{1}{d^2-1}\tr[\hUd_t\hU_L\hU_t\hUd_L]\tr[\hU_L\hUd_t\hUd_L\hU_t]
\end{align}
the first equation is gotten by re-arranging the traces and applying the trace property while the second equation is as a result of lemma \ref{lem:22design}.
\\\\
Equation \eqref{eqn:WAWB}  becomes
\begin{align}
\braket{(\partial_\nu C_g)^2}_{\hW_A, \hW_B, \hW_{A_1}}
=\frac{2(d^4+d^3-3d^2-2d+1)}{d^4(d^2-1)^3}+\frac{2}{d^4(d^2-1)^3}\tr[\hUd_t\hU_L\hU_t\hUd_L]\tr[\hU_L\hUd_t\hUd_L\hU_t]
\end{align}
Notice that taking expectation of the last equation with respect to $\hW_{B_1}$  does not change anything since the equation is independent on $\hW_{B_1}$.
\\\\
By a similar construction made on $\hU_{R}$, consider writing $\hU_L$ in terms of its last block i.e., let $\hU_L=\hU_{L_0}\hat{B_0}=\hU_{L_0}\hW_{A_0}\hW_{B_0}$. By using trace property,

\begin{equation}
\tr[\hUd_t\hU_L\hU_t\hUd_L]= \tr[\hW_{A_0}\hW_{B_0}\hU_t\hWd_{B_0}\hWd_{A_0}\hUd_{L_0}\hUd_t\hU_{L_0}], \;\;\;\;\tr[\hU_L\hUd_t\hUd_L\hU_t]=\tr[\hW_{A_0}\hW_{B_0}\hUd_t\hWd_{B_0}\hWd_{A_0}\hUd_{L_0}\hU_t\hU_{L_0}]
\end{equation}
Taking expectation of the above with respect to $\hW_{A_0}$ and applying lemma \ref{lem:2design}, we have
\begin{align}
\braket{\tr[\hUd_t\hU_L\hU_t\hUd_L]\tr[\hU_L\hUd_t\hUd_L\hU_t]}_{\hW_{A_0}}&= \frac{1}{d^2-1}(\tr[\hU_t]\tr[\hUd_t]\tr[\hUd_t]\tr[\hU_t]+d^2)\notag\\
&\hspace{2cm}-\frac{1}{d(d^2-1)}(d\tr[\hUd_t]\tr[\hU_t]+d\tr[\hU_t]\tr[\hUd_t])\\
&=\frac{1}{d^2-1}(|\tr[\hU_t]|^4 -2|\tr[\hU_t]|^2)+ \frac{d^2}{d^2-1}
\end{align}
Therefore,
\begin{align}\label{eq:GateCost1}
\braket{(\partial_\nu C_g)^2}_{\hU}=\frac{2(d^6+d^5-4d^4-3d^3+5d^2+2d-1)}{d^4(d^2-1)^4}+\frac{2}{d^4(d^2-1)^4}(|\tr[\hU_t]|^4-2|\tr[\hU_t]|^2)
\end{align}

\begin{equation}\label{eq:GateCost2}
\braket{(\partial_\nu C_g)^2}_{\hU} = \mathcal{O}(\mathrm{poly}(1/d^6))
\end{equation}
\end{proof}
Since all eigenvalues of unitaries are on the unit complex disk then the absolute value of the trace is bounded by the dimension $d$ of the qudit system; hence the equality.


%
%

\bibliographystyle{spmpsci}      
\bibliography{bibliography}   

%
%

\end{document}